\documentclass[runningheads,11pt]{llncs} 
\usepackage{geometry}
\usepackage[bookmarks=true]{hyperref}
\usepackage{amssymb}
\usepackage{etex} 
\usepackage{caption}
\usepackage{subcaption}

\usepackage[cmex10]{amsmath}

\usepackage{ntheorem}
\theoremstyle{break}

\usepackage{graphicx}
\usepackage{verbatim}
\usepackage{tikz}
\usetikzlibrary{arrows,positioning,shadows,matrix,calc}
\usepackage{pgf,pgfarrows,pgfnodes,pgfautomata,pgfheaps,pgfplots}
\usepackage{pst-plot}   

\usepackage[ruled,vlined,linesnumbered]{algorithm2e}

\usepackage{authblk}

\usepackage{wrapfig}

\usepackage{url}
\urldef{\mailsa}\path|{alfred.hofmann, ursula.barth, ingrid.haas, frank.holzwarth,|
\urldef{\mailsb}\path|anna.kramer, leonie.kunz, christine.reiss, nicole.sator,|
\urldef{\mailsc}\path|erika.siebert-cole, peter.strasser, lncs}@springer.com|

\begin{document}

\mainmatter  

\title{On the Information Privacy Model: the Group  and Composition Privacy}

\author{Genqiang Wu\footnote{This work was partially done when the author was a faculty member of the School of Information Engineering and was simultaneously a Ph.D candidate of the Institute of Software Chinese Academy of Sciences.}}
\institute{School of Computer Science and Engineering, Chongqing University of Technology 
\\
\email{wahaha2000@icloud.com}
}
\authorrunning{Genqiang Wu}

\maketitle

\begin{abstract}
How to query a dataset in the way of preserving the privacy of individuals whose data is included in the dataset is an important problem. 
The information privacy model, a variant of Shannon's information theoretic model to the encryption systems, protects the privacy of an individual by controlling the amount of information of the individual's data obtained by each adversary from the query's output. 
This model also assumes that each adversary's uncertainty to the queried dataset is not so small in order to improve the data utility.
In this paper, we prove some results to the group privacy and the composition privacy properties of this model, where the group privacy ensures a group of individuals' privacy is preserved, and where the composition privacy ensures multiple queries also preserve the privacy of an individual. Explicitly, we reduce the proof of the two properties to the estimation of the difference of two channel capacities.
Our proofs are greatly benefited from some information-theoretic tools and approaches. 
\end{abstract}


\newpage

\section{Introduction} \label{section-introduction}

Data privacy protection \cite{DBLP:journals/cacm/Dwork11,DBLP:journals/csur/FungWCY10,DBLP:series/ads/2008-34} studies how to query a dataset while preserving the privacy of individuals whose sensitive information is contained in the dataset. 
While our private data is increasingly collected and used, the data privacy protection becomes more and more important.  
Differential privacy model \cite{DBLP:conf/tcc/DworkMNS06,DBLP:conf/icalp/Dwork06} is currently the most important and popular data privacy protection model. This model has obtained great success in the computer science community\cite{DBLP:journals/fttcs/DworkR14}. 
It is very strong in protecting privacy than most other privacy models, such as $k$-anonymity\cite{DBLP:journals/ijufks/Sweene02}, $\ell$-diversity\cite{DBLP:conf/icde/MachanavajjhalaGKV06} or $t$-closeness\cite{DBLP:conf/icde/LiLV07} model etc.
It was even believed to be immune to any adversary with any background knowledge. However, with a decade or so of study, researchers found that the differential privacy model is vulnerable to those attacks where adversaries have knowledge of dependence among different records in the queried dataset\cite{DBLP:conf/sigmod/KiferM11}.    
Since then, many new privacy models are introduced to alleviate this defect of the differential privacy model, such as the Pufferfish model\cite{DBLP:journals/tods/KiferM14}, the coupled-worlds privacy model\cite{DBLP:conf/focs/BassilyGKS13}, the zero-knowledge privacy model\cite{DBLP:conf/tcc/GehrkeLP11}, the inferential privacy model\cite{DBLP:conf/innovations/GhoshK16} and the information privacy model\cite{wu-he-xia2018}, etc. Most of these models are inferential-based models; that is, an adversary's background knowledge is modeled as a probability distribution to the queried dataset and the aim of the adversary in general is to reduce the entropy of the probability distribution by analyzing the query results. These models  in general are more secure than the differential privacy model.
However, while the defect of the differential privacy model is removed, other problems arise: almost all of these models have a big weakness compared to the differential privacy model; that is, compared to the easily proved composition and group privacy properties of the differential privacy model, it is almost unknown if these models satisfy the two properties.
Note that the two properties are very important to each data privacy model, where the composition privacy property ensures that multiple queries also ensure each individual's privacy is protected and the group privacy property ensures that a group of individuals's privacy is also protected\cite{DBLP:journals/fttcs/DworkR14}.

In this paper we try to prove the group privacy and the composition privacy properties of the information privacy model\cite{wu-he-xia2018}. 
This model is a variant of Shannon's information-theoretic model to the encryption systems\cite{6769090}. 
Informally speaking, the information privacy model protect privacy by reducing the amount of disclosed information of each individual's data from query outputs obtained by each adversary.  
Furthermore, this model assumes that each adversary's uncertainty to the queried dataset is not so small.
This assumption is used to improve the data utility of the query's output, which is reasonable when the queried dataset is big enough.  One needs to be mentioned is that the differential privacy model is equivalent to the information privacy model when we assume that adversaries don't know the relationship among the records in the queried dataset\cite{wu-he-xia2018}. (Note that similar results also appear in many other papers, such as \cite{DBLP:journals/tods/KiferM14,DBLP:conf/ccs/LiQSWY13,DBLP:conf/innovations/GhoshK16}.)

\textbf{Our results} can be separated into two cases:
First, if there is no restriction on the adversaries' background knowledge, then the group privacy and the (basic) composition privacy properties of the information privacy model are almost the same with those of the differential privacy model, which are summarized as follows.
\begin{description}
\item [Group privacy:] We prove that if an algorithm satisfies $\epsilon$-information privacy, i.e. if each adversary can only obtain at most $\epsilon$ bits information of an individual from the algorithm's output, then the algorithm also ensures that each adversary can only obtain at most $k\epsilon$ bits information of a group of individuals  from the algorithm's output, where $k$ the size of the group.  
\item [Basic composition privacy:] We prove that if two algorithms satisfy respectively $\epsilon_1, \epsilon_2$-information privacy when they query \emph{the same dataset}, then the combination of the two algorithms' outputs also ensures that each adversary can only obtain at most $\epsilon_1+\epsilon_2$ bits information of an individual from the two outputs.
\item [General composition privacy:] We prove that if two algorithms satisfy respectively $\epsilon_1, \epsilon_2$-information privacy
 when they query respectively \emph{two different datasets}, then the combination of the two algorithms' outputs \emph{can't} ensure that each adversary can only obtain from the two outputs at most $\epsilon_1+\epsilon_2$ bits information of an individual. Note that this result is the first one appeared in the data privacy protection literatures. How to explain it is left as an open problem.
\end{description}
Second, 
if we assume that the lower bound of adversaries' uncertainties to the queried dataset is larger than zero,  
then the upper bounds of the disclosed information in the above three cases raise compared to the first case, where the amount of the increased information is positively related to the above lower bound of adversaries' uncertainties.
This phenomenon seems reasonable since the more larger restrictions to adversaries' knowledges, the more risky that the query outputs disclose information, especially in the scenarios of the group privacy and the composition privacy.

\section{The Information Privacy Model}

In this section, we restate the information privacy model introduced in \cite{wu-he-xia2018}. As mentioned in Section \ref{section-introduction}, this model is a variant of Shannon's information theoretic model to the encryption systems\cite{6769090}. In this paper most of the notations follow the book \cite{DBLP:books/daglib/0016881}. 

In an information privacy model there are $n\ge 1$ individuals and a dataset has $n$ records. 
Let the random variables $X_1, \ldots, X_n$ denote an adversary's probabilities/uncertainties to the $n$ records in the queried dataset. Let $\mathcal X_i$ denote the record universe of $X_i$.
A dataset $x:=(x_1,\ldots,x_n)$ is a sequence of $n$ records, where each $x_i$ is an assignment of $X_i$ and $x_i\in \mathcal X_i$.
Let $\mathcal X = \prod_{i\in [n]} \mathcal X_i$ where $[n]=\{1, \ldots, n\}$. Let $\mathbb P$ denote the universe of probability distributions over $\mathcal X$.
We abuse a capital letter, such as $X$, to either denote a random variable or denote the probability distribution which the random variable follows.  
Then if the probability distribution of the random variable $X$ is in $\Delta$, we say that $X$ is in $\Delta$, denoted as $X\in\Delta$.

\begin{definition}[The Knowledge of an Adversary]
Let the random vector $X:=(X_1, \ldots, X_n)$ denote the uncertainties/probabilities of an adversary to the queried dataset. Then $X$ or its probability distribution is called the knowledge of the adversary (to the dataset).
\end{definition}


In order to achieve more data utility, the information privacy model restricts adversaries' knowledges. 
Note that, by letting all adversaries' knowledges be derived from a subset $\Delta$ of $\mathbb P$, we achieve a restriction to adversaries' knowledges.
In this paper we assume that the entropy of each adversary's knowledge is not so small, which is formalized as the following assumption.

\begin{assumption} \label{assumption-1}
Let $b$ be a positive constant. Then, for any one adversary's knowledge $X$, there must be $X\in \mathbb P_b$, where
\begin{align}
\mathbb P_b = \{X: H(X)\ge b\}
\end{align}
with $H(X)$ being the entropy of $X$. 
\end{assumption}

For a query function $f$ over $\mathcal X$, let $\mathcal Y= \{f(x): x\in \mathcal X\}$ be its range. We now define the privacy channel/mechanism.
\begin{definition}[Privacy Channel/Mechanism]
To the function $f:\mathcal X \rightarrow \mathcal Y$,
we define a \emph{privacy channel/mechanism} to be a probability transition matrix $p(y|x)$ that expresses the probability of observing the output symbol $y\in \mathcal Y$ given that we query the dataset $x\in \mathcal X$.  
\end{definition}

To the above privacy channel $p(y|x)$, let $Y$ be its output random variable. We now define the individual channel capacity which is used to model the largest amount of information of an individual that an adversary can obtain from the output $Y$.  

\begin{definition}[Individual Channel Capacity]
To the function $f:\mathcal X \rightarrow \mathcal Y$ and its one privacy channel $p(y|x) $, we define the individual channel capacity of $p(y|x) $ with respect to $\Delta \subseteq \mathbb P$ as
\begin{align}
C_1=\max_{X\in \Delta, i\in [n]} I(X_i;Y),
\end{align}
where $X=(X_1,\ldots,X_n)$ and $I(X_i;Y)$ is the mutual information between $X_i$ and $Y$.
\end{definition}

The information privacy model protects privacy by controlling the individual channel capacity. 

\begin{definition}[Information Privacy] \label{definition-ip}
To the function $f:\mathcal X \rightarrow \mathcal Y$, we say that its one privacy channel $p(y|x) $ satisfies $\epsilon$-information privacy with respect to $\Delta$ if 
\begin{align}
C_1 \le \epsilon,
\end{align}
where $C_1$ is the individual channel capacity of $p(y|x)$ with respect to $\Delta$.
\end{definition}

If we set $\Delta=\mathbb P_b$ and assume that Assumption \ref{assumption-1} is true, then letting the channel $p(y|x)$ satisfy $\epsilon$-information privacy with respect to $\Delta$ will ensure that 
any adversary can only obtain at most $\epsilon$ bits information of each individual from the output of the channel.

We now define the (utility-privacy) balance function of a privacy channel. This function is used to express the ability of the privacy channel to preserve the privacy when adversaries' knowledges are reduced to $\mathbb P_b$ from $\mathbb P$ in order to improve the data utility. We stress that this function is very important in the following sections of the paper. 

\begin{definition}[Balance Function]  \label{definition-2}              
To the privacy channel $p(y|x)$, let $C_1, C_1^{b}$ be its individual channel capacities with respect to $\mathbb P, \mathbb P_b$, respectively. To each fixed $b$, assume there exists a nonnegative constant $\delta$ such that
\begin{align}
C_1 =C_1^b +\delta.
\end{align}
Then we say that the privacy channel is $(b,\delta)$-(utility-privacy) balanced, and that the function 
\begin{align}
\delta=\delta(b), \hspace{1cm} b\in [0,\log |\mathcal X|]
\end{align}
is the (utility-privacy) balance function of the privacy channel.
\end{definition}

The balance function is an increasing function. 

\begin{lemma}[Monotonicity of Balance Function] \label{lemma-4}
Let $\delta=\delta(b)$ be the balance function of the channel $p(y|x)$. Then the function is non-decreasing and therefore there is
\begin{align}
0=\delta(0)\le \delta(b)\le \delta(\log|\mathcal X|) < \min \left\{ b,\max_{i\in [n]}\log|\mathcal X_i| \right\}.
\end{align}
\end{lemma}

Clearly, the more larger $b$ is means that the more weaker adversaries are, which then implies the more larger data utility and the more larger $\delta$. Therefore, the parameter $b$ can be considered as an indicator of the data utility but the parameter $\delta$ can be considered as an indicator of the amount of the private information lost when the data utility has $b$ amount of increment.

\begin{theorem} \label{lemma-5}
Let $\delta=\delta(b)$ be the balance function of the privacy channel $p(y|x)$.
Then the privacy channel $p(y|x)$ satisfies $\epsilon$-information privacy with respect to $\mathbb P$ if and only if it satisfies $(\epsilon-\delta)$-information privacy with respect to $\mathbb P_b$.
\end{theorem}

Theorem \ref{lemma-5} is very useful since it turns hard problems into relatively easy problems. For example, if we want to construct an $\epsilon$-information privacy channel with respect to $\mathbb P_b$, we only need to construct an $(\epsilon+\delta)$-information privacy channel with respect to $\mathbb P$, where $\delta=\delta(b)$ is the balance function of the channel.
Furthermore, if the value of $\delta$ is relatively small, then Theorem \ref{lemma-5} implies that restricting the adversaries' knowledge from $\mathbb P$ to $\mathbb P_b$ doesn't significant change the model's ability to protect privacy even if Assumption \ref{assumption-1} is not true.

\section{Group Privacy}

The group privacy problem is to study whether privacy channels protect the privacy of a group of individuals. In this section we first formalize the group privacy definition and then prove a result to the group privacy.

\begin{definition}[Group Channel Capacity]
To the function $f:\mathcal X \rightarrow \mathcal Y$ and its one privacy channel $p(y|x) $, we define the $k$-group channel capacity of $p(y|x) $ with respect to $\Delta \subseteq \mathbb P$ as
\begin{align}
C_k=\max_{X\in \Delta, I\subseteq [n]: |I|=k} I(X_I;Y),
\end{align}
where $X_I=(X_{i_1},\ldots,X_{i_k})$ with $I=\{i_1,\ldots, i_k\}\subseteq [n]$.
\end{definition}

\begin{definition}[Group Privacy]
To the function $f:\mathcal X \rightarrow \mathcal Y$, assume its one privacy channel $p(y|x)$ satisfies $\epsilon$-information privacy with respect to $\Delta$.
We say that $p(y|x) $ satisfies $c$-group privacy with respect to $\Delta$ if 
\begin{align}
\max_{k\in [n]}C_k \le k(\epsilon+c),
\end{align}
where $C_k$ is the $k$-group channel capacity of $p(y|x)$ with respect to $\Delta$, $c$ is a nonnegative constant. 
\end{definition}

\begin{lemma} \label{lemma-2}
Let $X_I=(X_{I_1},X_{I_2})$ with $I=I_1\cup I_2$ and $I_1\cap I_2=\emptyset$. Then
there exists $|\mathcal X_{I_1}|$ probability distributions $\{X^{x_{I_1}}:x_{I_1}\in \mathcal X_{I_1}\}$ on $\mathcal X$ such that
\begin{align}
I(X_I;Y) = I(X_{I_1};Y) + \sum_{x_{I_1}} p(x_{I_1}) I(X_{I_2}^{x_{I_1}};Y^{x_{I_1}}),
\end{align}
where $X_{I_1} \sim p(x_{I_1})$, and
where each $Y^{x_{I_1}}$ is the corresponding output random variable of the channel $p(y|x)$ when the input random variable is $X^{x_{I_1}} \in \mathbb P$.
\end{lemma}

\begin{proof}
We only prove the case of $I_1=\{1\}, I_2=\{2\}$. Other cases can be proved similarly.

By the chain rule of the mutual information\cite[Theorem 2.5.2]{DBLP:books/daglib/0016881}, we have 
\begin{align}
I(X_I; Y) = I(X_1;Y)+I(X_2;Y|X_1)=I(X_1;Y)+\sum_{x_1} p(x_1) I(X_2;Y|x_1). 
\end{align}
Then, in order to prove the claim, to the fixed $x_1\in \mathcal X_1$ we only need to show that there exists a probability distribution $\tilde X \sim q(x)\in \mathbb P$ such that 
\begin{align}
I(\tilde X_2;\tilde Y)=I(X_2;Y|x_1),
\end{align}
where $\tilde Y$ is the corresponding output random variable of the channel $p(y|x)$.
To the fixed $x_1$, we construct the probability distribution $q(x)\in \mathbb P$ as follows. Let $q(x_1)$ be the probability distribution of $X_1$.
Set $q(x_{1})=1$ and $q(x_1')=0$ for $x_1'\in\mathcal X_1\setminus\{x_1\}$. 
Set $q(x_{(1)}|x_1)=p(x_{(1)}|x_1)  $ and $q(x) = q(x_1)q(x_{(1)}|x_1)$ for all $x=(x_1,x_{(1)}) \in \mathcal X$,
where $(i) = [n]-i$. Then it is easy to verify that $q(x)\in \mathbb P$. Let $\tilde X \sim q(x)$ and let $\tilde Y \sim q(y)$ be the output random variable when the source follows $\tilde X$ and when the channel is $p(y|x)$. Then
\begin{align}
q(y)=\sum_{x} q(x)p(y|x) =& \sum_{x_1'} q(x_1')\sum_{x_{(1)}}p(y|x_1',x_{(1)}) q(x_{(1)}|x_1')\\
=&\sum_{x_{(1)}}p(y|x_1,x_{(1)})  p(x_{(1)}|x_1) =p(y|x_1) 
\end{align}
and
\begin{align}
q(x_2,y) =& \sum_{x_{(2)}} p(y|x_2,x_{(2)})q(x_2,x_{(2)}) 
\\
=&\sum_{x_1'}q(x_1') \sum_{x_{(I)}} p(y|x_{1}',x_2,x_{(I)})q(x_{2},x_{(I)}|x_1') \\
=& \sum_{x_{(I)}} p(y|x_{1},x_2,x_{(I)})p(x_{2},x_{(I)}|x_1)  =p(x_2,y|x_1)
\end{align}
Then
\begin{align}
I(\tilde X_2;\tilde Y) =& \sum_{x_2,y}q(x_2,y)\log \frac{q(x_2,y)}{q(x_2)q(y)}\\
=& \sum_{x_2,y}p(x_2,y|x_1)\log\frac{p(y|x_1,x_2)}{p(y|x_1)}=I(X_2;Y|x_1).
\end{align}

The claim is proved.
\qed
\end{proof}


\begin{theorem}  \label{theorem-2}
Let the privacy channel $p(y|x)$ be $(b,\delta)$-balanced.
Then $p(y|x)$ satisfies $\delta$-group privacy  with respect to $\mathbb P_b$. 
\end{theorem}

\begin{proof}
The claim is a direct corollary of Lemma \ref{lemma-2} and Theorem \ref{lemma-5}.
\qed
\end{proof}

\section{Composition Privacy}

The composition privacy problem is to study whether the privacy channels protect each individual's privacy while multiple datasets or multiple query results are output.
There are two kinds of scenarios as discussed in \cite{wu-he-xia2018}.
 First, multiple query results of \emph{one dataset} are output. We call this kind of scenario as \emph{the basic composition privacy problem}. To the differential privacy model, the privacy problem of this scenario is treated by the \emph{composition privacy property}\cite{DBLP:journals/fttcs/DworkR14,DBLP:conf/sigmod/ZhangXX16,DBLP:journals/pvldb/ChenMFDX11,DBLP:conf/kdd/MohammedCFY11}. 
 
 Second, multiple query results of \emph{multiple datasets} generated by the same group of individuals are output, respectively. We call this kind of scenario as \emph{the general composition privacy problem}.
For example, the independent data publications of datasets of the Netflix and the IMDb \cite{DBLP:conf/sp/NarayananS08} respectively, the independent data publications of the online behaviors and offline retailing data \cite{DBLP:conf/kdd/LuoYLSYH16} respectively, and the independent data publications of the voter registration data and the medical data \cite{DBLP:journals/ijufks/Sweene02} respectively. For each of the above applications, the \emph{composition attack} \cite{DBLP:conf/kdd/GantaKS08,DBLP:conf/sp/NarayananS08} techniques may employ the relationship between/among different datasets/queries to infer the privacy of individuals whose data is contained in these datasets.

\subsection{Basic Composition Privacy}

\begin{definition}[Basic Composition Privacy] \label{definition-1}
To the function $f_j:\mathcal X \rightarrow \mathcal Y_j$, assume its one privacy channel $p(y_j|x)$ satisfies $\epsilon_j$-information privacy with respect to $\Delta$ for $j \in [m]$. 
We say that the composition channel $p(y|x)$ satisfies  $c$-(basic) composition privacy with respect to $\Delta$ if 
\begin{align}
C_1
\le \sum_{j\in [m]}\epsilon_j + c,
\end{align}
where
\begin{align}
p(y|x) = \prod_{j\in [m]}p(y_j|x),
\end{align}
$Y=(Y_1, \ldots, Y_m)$ and $y=(y_1,\ldots,y_m)$, and where $C_1$ is the individual channel capacity of $p(y|x)$ with respect to $\Delta$ and $c$ is a nonnegative constant. 
\end{definition}

\begin{lemma} \label{lemma-1}
Let the notations be as shown in  Definition \ref{definition-1} and set $m=2$.
Then there exist $|\mathcal Y_2|$ probability distributions $\{X^{y_2}:y_2\in \mathcal Y_2\}$ on $\mathcal X$ such that
\begin{align}
I(X_i;Y) = I(X_i;Y_2) + \sum_{y_2}p(y_2) I(X_i^{y_2};Y_1^{y_2}), i\in [n],
\end{align}
where each $Y_1^{y_2}$ is the output random variable of the channel $p(y|x)$ when the input random variable is $X^{y_2} \in \mathbb P$, and where
\begin{align}
p(y|x) = p(y_1|x)p(y_2|x).
\end{align}
\end{lemma}

\begin{proof}
The proof is similar to the proof of Lemma \ref{lemma-2}.
We only prove the case of $i=1$ and $n=2$, other cases can be proved similarly. 

By the chain rule of the mutual information, we have 
\begin{align}
I(X_1; Y) = I(X_1;Y_2)+I(X_1;Y_1|Y_2)=I(X_1;Y_2)+\sum_{y_2} p(y_2) I(X_1;Y_1|y_2). 
\end{align}
Then, in order to prove the claim, to the fixed $y_2$, we only need to show that there exists a probability distribution $\tilde X \sim q(x)\in \mathbb P$ such that 
\begin{align}
I(\tilde X_1;\tilde Y_1)=I(X_1;Y_1|y_2),
\end{align}
where $\tilde Y_1$ is the corresponding output random variable of the channel $p(y|x)$ when the input random variable is $\tilde X$.
To the fixed $y_2$, we construct the probability distribution $q(x)\in \mathbb P$ as follows.

Set $q(x) = q(x_1)q(x_2|x_1)$ with $q(x_1)=\frac{p(x_1)p(y_2|x_1)}{p(y_2)}=p(x_1|y_2), q(x_2|x_1)= \frac{p(y_2|x_1,x_2)p(x_2|x_1)}{p(y_2|x_{1})}$ and $q(y_1|x_1)=\sum_{x_2}p(y_1|x_1,x_2)q(x_2|x_1)$
 for all $x=(x_1,x_2) \in \mathcal X$. Then it is easy to verify that $q(x)\in \mathbb P$. Let $\tilde X$ follow $q(x)$ and let $\tilde Y_1 $ follow $ q(y_1)$, where $\tilde Y_1$ is the output random variable of the channel  $p(y_1|x)$ when the source is $\tilde X$. Then
\begin{align}
q(y_1)=\sum_{x} p(y_1|x)q(x) =& \sum_{x}p(y_1|x) \frac{p(y_2|x)p(x)}{p(y_2)}  \\
=&\frac{\sum_{x}p(y_1|x) p(y_2|x)p(x)}{p(y_2)} = p(y_1|y_2)
\end{align}
and
\begin{align}
q(x_1,y_1) =& \sum_{x_{2}} p(y_1|x_2,x_{1})q(x_2,x_{1})  \\
=&\sum_{x_{2}} p(y_1|x_2,x_{1})\frac{p(y_2|x_2,x_{1})p(x_2,x_{1})}{p(y_2)}  \\
=& \frac{\sum_{x_{2}} p(y_1|x_2,x_{1})p(y_2|x_2,x_{1})p(x_2,x_{1})}{p(y_2)}
= p(x_1,y_1|y_2).
\end{align}
Hence
\begin{align}
I(\tilde X_1;\tilde Y_1) =& \sum_{x_1,y_1}q(x_1,y_1)\log \frac{q(x_1,y_1)}{q(x_1)q(y_1)}\\
=& \sum_{x_1,y_1}p(y_1,x_1|y_2)\log \frac{p(y_1,x_1|y_2)}{p(x_1|y_2)p(y_1|y_2)}
=I(X_1;Y_1|y_2).
\end{align}

The claim is proved.
\qed
\end{proof}

\begin{theorem} \label{theorem-1}
Let the notations be as shown in  Definition \ref{definition-1}.
Let the privacy channel $p(y_j|x)$ be $(b,\delta_j)$-balanced, $j\in [m]$.
Then the composition channel $p(y|x)$ satisfies $\sum_{j\in [m]}\delta_j$-composition privacy  with respect to $\mathbb P_b$. 
\end{theorem}

\begin{proof}
The claim is a direct corollary of Lemma \ref{lemma-1} and Theorem \ref{lemma-5}.
\qed
\end{proof}

\subsection{General Composition Privacy}

We first formalize the general composition privacy problem. 
Set  $Y=(Y_1,\ldots,Y_m)$ and 
\begin{align}
X=(X_1,\ldots,X_n) = 
\begin{pmatrix}
X^1\\
\vdots \\
X^m
\end{pmatrix}
=\begin{pmatrix}
X_1^1 & \ldots & X_n^1 \\
\vdots & \ddots & \vdots \\
X_1^m & \ldots & X_n^m 
\end{pmatrix},
\end{align}
where 
\begin{align}
X_i =(X_i^1,\ldots, X_i^m)^T \mbox{\hspace{1cm} and \hspace{1cm}} X^j=(X_1^j, \ldots, X_n^j),
\end{align}
for $i\in [n], j \in [m]$.\footnote{One needs to be mentioned is that by convention the random variable $X_i^j$ should be denoted as $X_{ij}$. We choose  the former notation due to the notational abbreviation.} The notations $x_i, x^j, x_i^j, x, \mathcal X_i, \mathcal X^j, \mathcal X_i^j$ and $\mathcal X$ are defined accordingly. This implies that, to the general composition privacy problems, the input random variable $X$ is modeled as a random matrix, in which the column random vector $X_i$ is used to model an adversary's knowledge to the all $m$ records in the $m$ datasets of the $i$th individual and
the row random vector $X^j$ is used to model the adversary's knowledge to the $j$th dataset.  
Let $\mathbb P^j$ denote the universe of probability distributions over $\mathcal X^j$ for $j\in [m]$ and let $\mathbb P=\prod_{j\in [m]}\mathbb P^j$.

\begin{definition}[General Composition Privacy] \label{definition-general-composition}
To the function $f_j:\mathcal X^j \rightarrow \mathcal Y_j$, we assume its one privacy channel $p(y_j|x^j)$ satisfies $\epsilon_j$-information privacy with respect to $\Delta^j \subseteq \mathbb P^j$ for $j \in [m]$. 
We say that the composition channel $p(y|x) $ satisfies $c$-general composition privacy with respect to $\Delta=\prod_{j\in [m]}\Delta^j$ if 
\begin{align}
C_1 = \max_{X\in \Delta, i\in [n]}I(X_i;Y) \le \sum_{j\in [m]}\epsilon_j + c ,
\end{align}
where 
\begin{align}
p(y|x) =\prod_{j\in [m]} p(y_j|x^j),
\end{align}
and where $C_1$ is the individual channel capacity of $p(y|x)$ with respect to $\Delta$.
\end{definition}

\begin{lemma} \label{lemma-3}
Let $m=2$ and so $Y=(Y_1,Y_2)$ and $X=(X^1,X^2)$. 
Then there exist $|\mathcal X_1^1|+|\mathcal Y_1|+|\mathcal Y_1\times\mathcal X^1_1|$ probability distributions $\{X^{2,x_1^1}:x_1^1\in \mathcal X_1^1\}$, $\{X^{1,y_1} :y_1\in \mathcal Y_1\}$ and $\{X^{2,y_1,x^1_1}:y_1\in \mathcal Y_1,x^1_1 \in \mathcal X^1_1 \}$  such that
\begin{align}
I(X_1;Y) = I(X_1^1;Y_1) &+ \sum_{x_1^1}p(x_1^1) I(X_1^{2,x_1^1};Y_1^{x_1^1}) + \sum_{y_1}p(y_1)I(X_1^{1,y_1};Y_2^{y_1}) \\
&+\sum_{y_1}p(y_1)\sum_{x_1^1}p(x_1^1)I(X_1^{2,x_1^1,y_1};Y_2^{x_1^1,y_1})  
\end{align}
where each $Y_1^{x_1^1}$ is the corresponding output random variable of the channel $p(y_1|x^2)$ when the input random variable is $X^{2,x_1^1} \in \mathbb P^2$, where each $Y_2^{x_1^1,y_1}$ is the corresponding output random variable of the channel $p(y_2|x^2)$ when the input random variable is $X^{2,x_1^1,y_1} \in \mathbb P^2$, 
and where each $Y_2^{y_1}$ is the corresponding output random variable of the channel  $p(y_2|x^1)  $ when the input random variable is $X^{1,y_1} \in \mathbb P^1$.
\end{lemma}

\begin{proof}
The proof is a combination of the proofs of Lemma \ref{lemma-1} and Lemma \ref{lemma-2}. We have
\begin{align*}
&I(X_1;Y) \\
=& I(X_1;Y_1) + I(X_1;Y_2|Y_1) \\
=_a& I(X_1^1;Y_1) + I(X_1^2;Y_1|X_1^1) + \sum_{y_1}p(y_1)I(X_1^{y_1};Y_2^{y_1})  \\
=_b& I(X_1^1;Y_1) + I(X_1^2;Y_1|X_1^1) + \sum_{y_1}p(y_1)I(X_1^{1,y_1};Y_2^{y_1}) 
+\sum_{y_1}p(y_1)\sum_{x_1^1}p(x_1^1)I(X_1^{2,x_1^1,y_1};Y_2^{x_1^1,y_1}) \\
=_c& I(X_1^1;Y_1) + \sum_{x_1^1}p(x_1^1) I(X_1^{2,x_1^1};Y_1^{x_1^1}) + \sum_{y_1}p(y_1)I(X_1^{1,y_1};Y_2^{y_1}) 
+\sum_{y_1}p(y_1)\sum_{x_1^1}p(x_1^1)I(X_1^{2,x_1^1,y_1};Y_2^{x_1^1,y_1}) , 
\end{align*}
where the equalities $=_a, =_b$ and $=_c$ are due to Lemma \ref{lemma-2} and Lemma \ref{lemma-1}.

The claim is proved. 
\qed
\end{proof}

\begin{theorem} \label{theorem-3}

To each channel $p(y_i|x^j)$, assume $C_{i}^{\mathbb P_b^j}$ be its individual channel capacity with respect to $\mathbb P_b^j$, $i,j\in [2]$. 
Assume the channel $p(y_j|x^j)$  is $(b,\delta_j)$-bounded, $j\in [2]$. 
Then,  the composition channel $p(y|x)$ satisfies $(\sum_{j\in [2]}\delta_j+C_{1}^{\mathbb P_b^2}+C_{2}^{\mathbb P_b^1})$-general composition privacy with respect to $\prod_{j\in [2]}\mathbb P_b^j$. 
\end{theorem}

\begin{proof}
The claim is a direct corollary of Lemma \ref{lemma-3} and Theorem \ref{lemma-5}.
\qed
\end{proof}

Note that the claims of Lemma \ref{lemma-3} and Theorem \ref{theorem-3} can be generalized to the cases of $m\ge 2$, whose  proofs are similar to the proofs of Lemma \ref{lemma-3}  and Theorem \ref{theorem-3}. Furthermore, one needs to be stressed is that the parameter 
\begin{align}
c=\sum_{j\in [2]}\delta_j+C_{1}^{\mathbb P_b^2}+C_{2}^{\mathbb P_b^1}
\end{align}
in Theorem \ref{theorem-3} can't be reduced to $0$ when $b \rightarrow 0$ whereas the case in Theorem \ref{theorem-1} is $c \rightarrow 0$ when $b \rightarrow 0$. 
We don't know how to explain this phenomenon; that is, we don't know whether this phenomenon implies that the information privacy model is not secure. Our belief is that it doesn't.
Of course, since the result in Theorem \ref{theorem-3} is a new thing, further results and how to explain these results need to be explored. 
Furthermore, it seems that Theorem \ref{theorem-3} is related to the continue observation problems\cite{DBLP:conf/stoc/DworkNPR10}.

\section{Conclusion}

In this paper we proved the group privacy and the composition privacy properties of the information privacy model. Although our results contain the unknown balance functions and then we can't seem to use these results, there in fact exists a way to use them due to the monotonicity of the balance functions.
Specifically, we first set the needed value of $\delta$ (but not $b$), then there must \emph{exist} a largest $b$ such that the equation $\delta(b) \le \delta$ (even though we can't find it). Then we can freely use all the four theorems of this paper. The only drawback of this way is that we don't know the value of $b$ and then can't estimate the extent of reasonability of Assumption \ref{assumption-1}.

This paper leaves several unsolved problems. First, 
our results to the group privacy and the composition privacy are related to the balance function of the corresponding privacy channel, which seems to be rather complex than those of the differential privacy model. Whether the balance function can be deleted from these results is unknown. Second, how to evaluate the balance function of each privacy channel and whether the balance function is monotonic to the number of individuals $n$ are interesting problems.
Third,  how to explain the result of the general composition privacy is an important problem.

\bibliographystyle{unsrt}

\end{document}